\documentclass[american,aps,pra, reprint]{revtex4-1}
\usepackage[T1]{fontenc}
\usepackage[latin9]{inputenc}
\usepackage{babel}
\usepackage{amsthm}
\usepackage{amsmath}
\usepackage{bm}
\usepackage{amssymb}
\usepackage{graphics,graphicx}
\usepackage{tikz}
\usepackage{float}
\usepackage{multirow}
\usepackage[unicode=true,pdfusetitle, bookmarks=true,bookmarksnumbered=false,bookmarksopen=false,
 breaklinks=false,pdfborder={0 0 0},backref=false,colorlinks=false]
 {hyperref}
\hypersetup{colorlinks,linkcolor=myurlcolor,citecolor=myurlcolor,urlcolor=myurlcolor}

\makeatletter
\@ifundefined{textcolor}{}
{%
 \definecolor{BLACK}{gray}{0}
 \definecolor{WHITE}{gray}{1}
 \definecolor{RED}{rgb}{1,0,0}
 \definecolor{GREEN}{rgb}{0,1,0}
 \definecolor{BLUE}{rgb}{0,0,1}
 \definecolor{CYAN}{cmyk}{1,0,0,0}
 \definecolor{MAGENTA}{cmyk}{0,1,0,0}
 \definecolor{YELLOW}{cmyk}{0,0,1,0}
 }
\theoremstyle{plain}
\newtheorem{thm}{\protect\theoremname}
\ifx\proof\undefined
\newenvironment{proof}[1][\protect\proofname]{\par
\normalfont\topsep6\p@\@plus6\p@\relax
\trivlist
\itemindent\parindent
\item[\hskip\labelsep
\scshape
#1]\ignorespaces
}{%
\endtrivlist\@endpefalse
}
\providecommand{\proofname}{Proof}
\fi

\usepackage{times}
\usepackage{txfonts}
\usepackage{braket}
\usepackage{colortbl}
\definecolor{myurlcolor}{rgb}{0,0,0.7}

\makeatother

\providecommand{\theoremname}{Theorem}

\usepackage{times}

\usepackage[up]{subfigure}

\begin{document}
%\title{Reciprocity between Quantum Coherence and mixedness: a general feature}
\title{Quantum Coherence: Reciprocity and Distribution}
\author{Asutosh Kumar}
%\author{Asutosh Kumar, Aditi Sen (De), Ujjwal Sen}
\email{asukumar@hri.res.in}
 
\affiliation{%Quantum Information and Computation Group,\\
Harish-Chandra Research Institute, Allahabad-211019, India}
\affiliation{Homi Bhabha National Institute, Anushaktinagar, Mumbai 400094, India}

\begin{abstract}
Quantum coherence is the outcome of the superposition principle. Recently, it has been theorized as a quantum resource, and is the premise of quantum correlations in multipartite systems. 
%It is the premise of quantum correlations in multipartite systems, and has been vital in accomplishing tasks beyond the classical realm. 
%However, coherence is degraded by environmental interactions. 
It is therefore interesting to study the coherence content and its distribution in a multipartite quantum system.
In this work, we show analytically as well as numerically the reciprocity between coherence and mixedness of a quantum state. We find that this trade-off is a general feature in the sense that it is true for large spectra of measures of coherence and of mixedness.
We also study the distribution of coherence in multipartite systems by looking at monogamy-type relation--which we refer to as additivity relation--between coherences of different parts of the system.  
%%and prove several interesting results. 
%Numerical investigation shows that the percentage of quantum states satisfying the additivity relation of coherence increases with increasing number of parties, with increment in the rank of quantum states, and with raising of the power of coherence measures under investigation. 
%%We prove analytically that a measure of coherence that satisfies the additivity relation does satisfy the same on raising its power, and a measure of coherence that violates the additivity relation does violate the same on lowering its power. We also study the distribution of coherence in multipartite X states. 
We show that for the Dicke states, while the normalized measures of coherence violate the additivity relation, the unnormalized ones satisfy the same.  
\end{abstract}

\maketitle

\section{Introduction}
\label{intro}
From everyday life experiences, we learn that arbitrary operations cannot do an assigned job. That is, specific resources--like allowed operations, ``free assets'' that one can use at will, and some ``force'' or catalyst in a prescribed amount--are needed to carry out a particular task. Therefore, to establish a quantitative theory of any physical resource, one needs to address the following fundamental issues: (i) the characterization or unambiguous definition of resource, (ii) the quantization or valid measures, and (iii) the transformation or manipulation of quantum states under the imposed constraints \cite{const1,const2,const3,const4}. Several useful quantum resources like purity \cite{purity}, entanglement \cite{ent1,ent2,ent3,ent4,ent5}, reference frames \cite{ref-frame1,ref-frame2}, thermodynamics \cite{thmdy1,thmdy2}, etc. have been identified and quantified until now. Recently, Baumgratz \emph{et al.} in Ref. \cite{baumgratz}, provided a quantitative theory of coherence as a new quantum resource, borrowing the formalism already established for entanglement \cite{ent1,ent2,ent3,ent4,ent5}, thermodynamics \cite{thmdy1,thmdy2} and reference frames \cite{ref-frame1,ref-frame2}.

%Quantum mechanics identifies superposition of possible eigenstates as a valid quantum state.
Coherence arises from the superposition principle, and is defined for single as well as multipartite systems.
%aided with the tensor product structure of the Hilbert space. 
%In optics, coherence is captured in interference results. 
Quantum coherence is identified by the presence of off-diagonal terms in the density matrix, and hence is a basis-dependent quantity. It being a basis-dependent quantity, local and nonlocal unitary operations can alter the amount of coherence in a quantum system. 
A density matrix has zero coherence with respect to a specific basis if it is diagonal in that basis.
Diagonal density matrices, in the above sense, therefore represent essentially the classical mixtures. 
A coherent quantum state is considered as a resource in thermodynamics as it allows non-trivial transformations \cite{coh-resource}.
%Basically, coherence describes the correlation ...mother of all correlations in quantum theory.   
Quantum superposition is the most fundamental feature of quantum mechanics. Quantum coherence is a direct consequence of the superposition principle. Moreover, combined with the tensor product structure of quantum state space, it gives rise to the novel concepts such as entanglement and quantum correlations.
It, being the premise of quantum correlations in multipartite systems, has attracted the attention of 
quantum information community significantly, and in addition to its quantification \cite{coh-measure1,coh-measure2,coh-measure3}, other developments like the freezing phenomena \cite{freezing}, the coherence transformations under incoherent operations \cite{coh-trans}, establishment of geometric lower bound for a coherence measure \cite{diogo}, the complementarity between coherence and mixedness \cite{complement}, its relation with other measures of quantum correlations and creation of coherence using unitary operations \cite{yao, coh-avijit}, erasure of quantum coherence \cite{erasing-uttam}, and catalytic transformations of coherence \cite{catalytic-uttam} have been reported recently.

In this paper, we revisit the complementarity between coherence and mixedness of a quantum state, and the distribution of coherence in multipartite systems in considerable detail. We provide analytical and numerical results in this regard. 
This paper is organized as follows. In Sec. \ref{meas}, we briefly define the measures that quantify quantum coherence and mixedness. In Sec. \ref{tradeoff}, we show that the reciprocity between coherence and mixedness in quantum systems is an extensive feature in the sense that it holds for large spectra of measures of coherence and of mixedness. In Sec. \ref{coh-mono}, we discuss the distribution of coherence in multipartite quantum systems. 
Numerical investigation unravels the fact that the percentage of quantum states satisfying the additivity relation of coherence increases with increasing number of parties, with increment in the rank of quantum states, and with raising of the power of coherence measures under investigation.
We provide conditions for the violation of the additivity relation of the relative entropy of coherence.
%and prove analytically that a measure of coherence that satisfies the additivity relation does satisfy the same on raising its power, and a measure of coherence that violates the additivity relation does violate the same on lowering its power. 
In Sec. \ref{x-state}, we investigate the distribution of coherence in a special type of quantum states called ``X''-states, and provide examples. Finally, we conclude our findings in Sec. \ref{discussion}.

\section{Quantifying coherence and mixedness of a quantum state}
\label{meas}
In this section, we briefly review the axiomatic approach to characterize and quantify coherence, as proposed in Ref. \cite{baumgratz}, and mixedness of a quantum system.

%%We briefly introduce the formalism to characterize and quantify coherence and mixedness of quantum systems, which are the two central quantities in our present investigation. 
%%
%In this section we present a brief overview of the concepts of quantum coherence and mixedness of quantum systems. To characterize the coherence in a quantum system, we follow the theoretical approach developed in Ref.~\cite{Baumgratz2014}. 
%All mathematical formulations and results that are subsequently presented and discussed are valid within the framework of the above theory of quantum coherence.

%\noindent{\it Quantum coherence}-- 
\subsection{Quantum coherence} 
In the framework of Ref. \cite{baumgratz}, all the diagonal states, in a given reference basis, constitute a set of incoherent states, denoted by $\mathcal{I}$. 
%And incoherent operations, $\Phi_{\mathcal{I}}$, are completely positive trace preserving (CPTP) maps which transform the set of incoherent states onto itself, i.e., $\Phi_\mathcal{I}(\mathcal{I}) \in \mathcal{I}$. 
And a completely positive trace preserving (CPTP) map is an incoherent operation if it possesses a Kraus operator decomposition $\{K_t\}$ such that $K_t\rho K_t^{\dagger}$ is incoherent for every incoherent state $\rho \in \mathcal{I}$.
%Under the set of operations $\Phi_\mathcal{I}$ and the free incoherent states $\mathcal{I}$, quantum coherence is a valid resource that can be quantified. 
%
A function, $C(\rho)$, is a valid measure of quantum coherence of the state $\rho$ if it satisfies the following conditions \cite{baumgratz}:
(1) $C(\rho) = 0$ iff $\rho \in \mathcal{I}$.
(2a) Monotonicity under the incoherent operations, i.e., $C\left(\Phi_I(\rho)\right) \leq C(\rho)$.
(2b) Monotonicity under the selective incoherent operations on an average, i.e., $\sum _k p_k  C(\rho_k) \leq C(\rho)$, where $\rho_k = M_k \rho M^\dagger_k/p_k$, $p_k = \mathrm{Tr}M_k \rho M^\dagger_k$, and $M_k$ are the ``incoherent'' Kraus operators as described above. That is, $C(\rho)$ is non-increasing on an average under the selective incoherent operations.
(3) Convexity or nonincreasing under mixing of quantum states, i.e., $C(\sum_k p_k \rho_k) \leq \sum_k p_k C(\rho_k)$. That is, coherence cannot increase under mixing. 

It is emphasized that the incoherency condition, $M \mathcal{I}M^{\dagger} \in \mathcal{I}$, places a severe constraint on the structure of the incoherent Kraus operator \(M\) \cite{yao}: there can be at most one nonzero entry in every column of \(M\). Thus, if the incoherent Kraus operator $M$ belongs to the set of $r \times c$  matrices ${\cal M}_{r,c}$, then the maximum number of possible structures of $M$ is $r^c$.
Note further that conditions (3) and (2b) together imply condition (2a) \cite{baumgratz}:
\begin{align}
\label{eq:2b32a}
C\left(\Phi_{I}(\rho)\right)=C\left(\sum_n p_n \rho_n \right) \overset{(3)}{\leq} \sum_n p_n C(\rho_n) \overset{(2b)}{\leq} C(\rho).
\end{align}

Measures that satisfy the above conditions, include 
% Based on these bona fide conditions one can define various functions that will quantify quantum coherence in a given quantum system. The examples of quantum coherence measures include, 
$l_1$ norm and relative entropy of coherence \cite{baumgratz} and the skew information \cite{coh-measure1}. Coherence can also be quantified through entanglement. It was shown in Ref. \cite{coh-measure2} that entanglement measures which satisfy above conditions can be used to derive generic monotones of quantum coherence. 
Recall that quantum coherence is a {\it basis-dependent} quantity. Yao {\it et al.} in \cite{yao} asked whether a basis-independent measure of quantum coherence can be defined. 
They observed that the {\it basis-free} coherence is equivalent to quantum discord \cite{yao}, supporting the fact that coherence is a form of quantum correlation in multipartite quantum systems.
Viewing a $d$-dimensional quantum state $\rho$, in the reference basis $\{\ket{i}\}$, as a $d^2$-dimensional vector, its $l_p$ norm is 
\begin{align}
\|\rho \|_p = \left(\sum_{i,j} |\rho_{ij}|^p\right)^{\frac{1}{p}},
\end{align}
where $\rho_{ij}=\langle i | \rho | j\rangle$.
The quantity $C_{l_1}(\rho)$, which is based on $l_1$ norm, and given by
\begin{align}
 C_{l_1}(\rho) = \sum_{i\neq j} |\rho_{ij}|,
\end{align}
is a valid measure of coherence \cite{baumgratz}.
Another quantity, $C_{r}(\rho) = \min_{\sigma \in \mathcal{I}} S(\rho \parallel \sigma) = S(\rho_I) - S(\rho)$, is the relative entropy of coherence,
where $\mathcal{I}$ is the set of incoherent states in the reference basis, $S(\rho \parallel \sigma) = \mbox{Tr} \rho(\log\rho - \log\sigma)$ is the relative entropy between $\rho$ and $\sigma$, and $\rho_I=\sum_i \langle i | \rho | i\rangle | i\rangle\langle i |$. Furthermore, a geometric measure of coherence is also proposed \cite{baumgratz, coh-measure2, coh-measure3} which is a full coherence monotone \cite{coh-measure2}. The geometric measure is given by $C_g(\rho) = 1-\max_{\sigma\in\mathcal{I}} F(\rho,\sigma)$, where $\mathcal{I}$ is the set of all incoherent states and $F(\rho,\sigma) = \left(\mathrm{Tr}\left[\sqrt{\sqrt{\sigma}\rho\sqrt{\sigma}}\right]\right)^2$ is the fidelity \cite{nielsen00} of the states $\rho$ and $\sigma$.
%It is important to note that quantum coherence, by definition, is not invariant under general unitary operation but does remain unchanged under incoherent unitaries. 
The maximally coherent pure state is defined by $\ket{\psi_d} = \frac{1}{\sqrt{d}}\sum_{i=0}^{d-1} \ket{i}$ \cite{baumgratz}, for which $C_{l_1}(\ket{\psi_d}\bra{\psi_d}) = d-1$ and $C_r(\ket{\psi_d}\bra{\psi_d}) = \mbox{ln}~d$. 

%\vspace{0.1cm}
%\noindent\emph{Mixedness}--
\subsection{Mixedness}
A quantum system which is not pure is mixed. A pure quantum system described by density matrix $\rho$ is characterized by $\mbox{Tr}(\rho^2)=\mbox{Tr}(\rho)=1$. $\mbox{Tr}(\rho^2)$ is called the purity of $\rho$. Noise in various forms, including inevitable interaction with environment, degrades the purity of a quantum state and renders it mixed. Mixedness characterizes disorder or loss of information, and is a complementary quantity to the purity of a quantum system.
%%% 
There are several ways to quantify the mixedness of a quantum state in the literature. For an arbitrary $d$-dimensional state, the mixedness, based on normalized linear entropy \cite{norm-linear-ent}, is given as  
\begin{align}
 M_l(\rho) = \frac{d}{d-1}\left( 1 - \mathrm{Tr}\rho^2 \right).
\end{align}
Therefore, for every quantum system, mixedness varies between zero and unity. 
%Furthermore, since $\mathrm{Tr}\rho^2$ describes the purity of quantum system, mixedness expectedly emerges as a complementary quantity to the purity of the given quantum state. 
The other operational measures of mixedness of a quantum state $\rho$ include the von Neumann entropy $S(\rho)=-\mbox{Tr}(\rho\mbox{ln} \rho)$, and geometric measure of mixedness which is given by $M_g(\rho) := F(\rho,\mathbb{I}/d) = \frac{1}{d}\left( \mathrm{Tr}\sqrt{\rho} \right)^2$.  %and lies between $0$ and $1$. 
For a $d$-dimensional pure quantum states $\ket{\phi_d}$, while $M_l(\ket{\phi_d})$ and $S(\ket{\phi_d})$ vanish, $M_g(\ket{\phi_d})=\frac{1}{d}$. Thus, $M_l(\ket{\phi_d})$ and $S(\ket{\phi_d})$ lie between $0$ and $1$, and
$M_g(\ket{\phi_d})$ varies between $\frac{1}{d}$ and $1$.

\section{Reciprocity between quantum coherence and mixedness}
\label{tradeoff}

%\noindent\emph{Trade-off between quantum coherence and mixedness}.--
As mixedness is complementary to purity and purity is closely related to quantum coherence, it is natural to 
investigate the restrictions imposed by the mixedness of a system on its quantum coherence. In this section, we show analytically and numerically that there exists a trade-off between the two quantities for different measures of coherence and mixedness.\\
For any arbitrary quantum system $\rho$ in $d$ dimensions, quantum coherence, as quantified by the $l_1$ norm, and mixedness, in terms of the normalized linear entropy, satisfies the following inequality
\begin{align}
\label{eq:comp-l1}
\frac{C_{l_1}^2(\rho)}{(d-1)^2} + M_l(\rho) \leq 1.
\end{align}

Inequality (\ref{eq:comp-l1}) dictates that for a fixed amount of mixedness, the maximal amount of coherence is limited, and vice-versa.
This important trade-off relation between quantum coherence and mixedness was obtained in Ref. \cite{complement} using the parametric form of an arbitrary $d$-dimensional density matrix, written in terms of the generators, ${\cal G}_i$, of $SU(d)$ \cite{nielsen00, eberly81, mahler98, kimura03, khaneja03}, as
\begin{align}
\label{qd}
\rho = \frac{\mathbb{I}}{d} + \frac{1}{2}\vec{x}.\vec{{\cal G}} = \frac{\mathbb{I}}{d} + \frac{1}{2}\sum_{i=1}^{d^2-1}x_i {\cal G}_i,
\end{align}
where $x_i = \mathrm{Tr}[\rho {\cal G}_i]$. The generators ${\cal G}_i$ satisfy $(i)$  ${\cal G}_{i} = {\cal G}_{i}^\dag$, $(ii)$ $\mbox{Tr}({\cal G}_{i})=0$, and $(iii)$ $\mbox{Tr}({\cal G}_{i}{\cal G}_{j})=2\delta_{ij}$. In this representation, three-dimensional state is
\begin{eqnarray}
\label{qs3}
\rho=\left(
\begin{array}{ccc}
\frac13+x_7+\frac{x_8}{\sqrt{3}} & x_1-ix_4 & x_2-ix_5\\
x_1+ix_4 & \frac13-x_7+\frac{x_8}{\sqrt{3}} & x_3-ix_6\\
x_2+ix_5 & x_3+ix_6 & \frac13-\frac{2x_8}{\sqrt{3}}
\end{array}
\right).
\end{eqnarray}

The $l_1$ norm of coherence of a $d$-dimensional system, given by  Eq. (\ref{qd}), can be written as \cite{complement}
\begin{align}
 C_{l_1}(\rho) 
 %&= \sum_{\substack{m,n = 1\\m\neq n}}^{d} |\rho_{mn}|
 %= \frac{1}{2}\sum_{\substack{m,n = 1\\m\neq n}}^{d} |\sum_{i=1}^{d^2-1}x_i\bra{m} \hat{\Lambda}_i %\ket{n}| \nonumber\\
 & = \sum_{i=1}^{(d^2-d)/2}\sqrt{x_i^2 + x_{i+(d^2-d)/2}^2 }.
\label{l1eq}
\end{align}
And, the mixedness is given by
\begin{align}
 M_l(\rho) = 1 - \frac{d}{2(d-1)}\sum_{i=1}^{d^2-1}x_i^2.
\label{mixeq}
\end{align}

Eq. (\ref{eq:comp-l1}) ensures that the normalized coherence, $\frac{C_{l_1}(\rho)}{(d-1)}$, of a quantum system with mixedness $M_l(\rho)$, is bounded to a region below the ellipse $\frac{C_{l_1}^2(\rho)}{(d-1)^2} + \left(\sqrt{M_l(\rho)}\right)^2 = 1$. The quantum states with (normalized) quantum coherence that lie on the conic section are the maximally coherent states corresponding to a fixed mixedness and vice-versa \cite{complement}.\\

It is interesting to note that provided $C_{l_2}(\rho) = \left(\sum_{i\neq j} |\rho_{ij}|^2\right)^\frac12$ were a valid coherence measure, one could easily show that a complementarity relation, analogous to Eq. (\ref{eq:comp-l1}), holds: 
\begin{align}
\label{eq:comp-l2}
\frac{C_{l_2}^2(\rho)}{\left(\sqrt{1-\frac{1}{d}}\right)^2} + M_l(\rho) \leq 1.
\end{align}
In Ref. \cite{baumgratz}, it was shown that the quantity 
$\tilde{C}_{l_2}(\rho)=C_{l_2}^2(\rho)=\sum_{i\neq j} |\rho_{ij}|^2$ satisfies conditions (1) and (3). However, it fails to satisfy the condition (2b) in general. 
%Since conditions (2b) and (3) together imply condition (2a) (see Eq. (\ref{eq:2b32a})), it does not satisfy condition (2a) in general. 
Thus it is not clear whether $C_{l_2}^2(\rho)$ is a valid coherence measure in the framework of above resource theory.  

A natural question that arises is whether the reciprocity between quantum coherence and mixedness is measure specific? Put differently, does complementarity between coherence and mixedness hold for other measures of coherence and mixedness. 
%And the answer is in affirmative. 
It is trivial to note that
\begin{align}
\label{eq:comp-ent}
\frac{C_{r}(\rho)}{\mbox{ln}~d} + \frac{S(\rho)}{\mbox{ln}~d} \leq 1,
\end{align}
and
\begin{align}
\label{eq:comp-geom}
C_{g}(\rho) + M_g(\rho) &= 1-\left(\max_{\sigma\in\mathcal{I}} F(\rho,\sigma)-F(\rho, \mathbb{I}/d)\right) \nonumber \\
& \leq 1.
\end{align}
We observe from Eqs. (\ref{eq:comp-l1}), (\ref{eq:comp-ent}) and (\ref{eq:comp-geom}) that for valid coherence measures, there is trade-off between functions of normalized coherence and normalized mixedness. This complementarity between coherence and mixedness appears to be a general feature. It would be an interesting exercise to investigate whether a given measure of coherence respects reciprocity with different measures of mixedness. In particular, we are interested in whether the following relations hold:
\begin{align}
\frac{C_{l_1}^2(\rho)}{(d-1)^2} + \frac{S(\rho)}{\mbox{ln}~d} & \leq 1 \nonumber \\
\frac{C_{r}(\rho)}{\mbox{ln}~d} + M_l(\rho) & \leq 1, \ \ \mbox{etc.}
\label{eq:comp-gen}
\end{align}
For rank-1 (pure) states, the above complementarity relations are trivially satisfied since mixedness for pure states is, by definition, zero (for geometric measure of coherence, one will have to set 
$M_g(|\psi\rangle)=0$ by hand). Interestingly, for higher rank quantum states also the above reciprocity relations hold. We provide numerical evidences which suggest that trade-off between coherence and mixedness is indeed an extensive feature of quantum systems (see Figs. \ref{fig:coh-mix-reciprocity-3q} and \ref{fig:coh-mix-reciprocity-4q}). Though the reciprocity relation, $\frac{C_{r}(\rho)}{\mbox{ln}~d} + M_l(\rho) {\leq} 1$, is in conflict, it is well below the trivial value 2, for all states. We observe that higher is the rank of quantum states and number of qubits, more is the violation. We found numerically that the reciprocity relation, $\frac{C_{r}(\rho)}{\mbox{ln}~d} + M_l(\rho) {\leq} 1$, is violated by two-qubit states also. An example of a two-qubit state which violates this relation is given in Eq. (\ref{2qubit-violation}).

\begin{widetext}
\begin{eqnarray}
\label{2qubit-violation}
\rho=\left(
\begin{array}{cccc}
0.2501 & 0.0490-0.0090 i & -0.1392-0.1148 i & -0.2141-0.0515 i\\
0.0490+0.0090 i & 0.2064 & 0.1588-0.0438 i & 0.0137+0.0650 i\\
-0.1392+0.1148 i & 0.1588+0.0438 i & 0.3001 & 0.1858+0.0115 i\\
-0.2141+0.0515 i & 0.0137-0.0650 i & 0.1858-0.0115 i & 0.2434
\end{array}
\right).
\end{eqnarray}
Note that $\rho = \rho^{\dagger}$, $Tr \rho = 1$, $Tr \rho^2 = 0.5539$, and eigenvalues of $\rho$ are $\{0.664,~0.336,~0,~0\}$. Hence, $\rho$ is a valid rank-2 density matrix. For this density matrix,
%\begin{equation}
$\frac{C_{r}(\rho)}{2} + M_l(\rho)=0.5334+0.5948=1.1282 > 1$. 
%\end{equation}
\end{widetext}

%%%
\begin{center}
\begin{figure}[htb]
%\subfigure[]{\includegraphics[width=1.2in, angle=0]{3a}} \hspace{0.2cm}
\includegraphics[width=1.6in, angle=0]{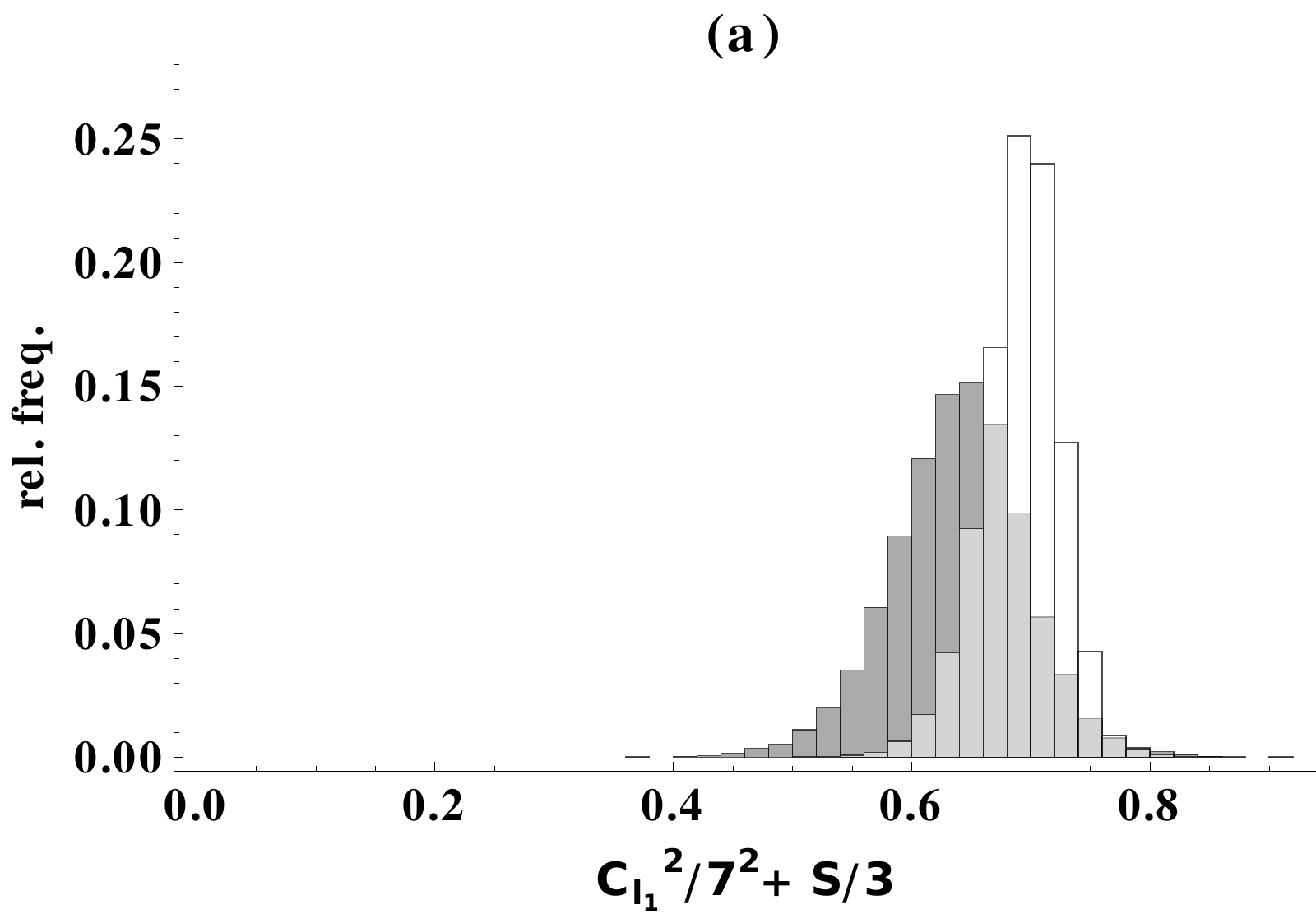} 
%\hspace{0.1cm}
\includegraphics[width=1.6in, angle=0]{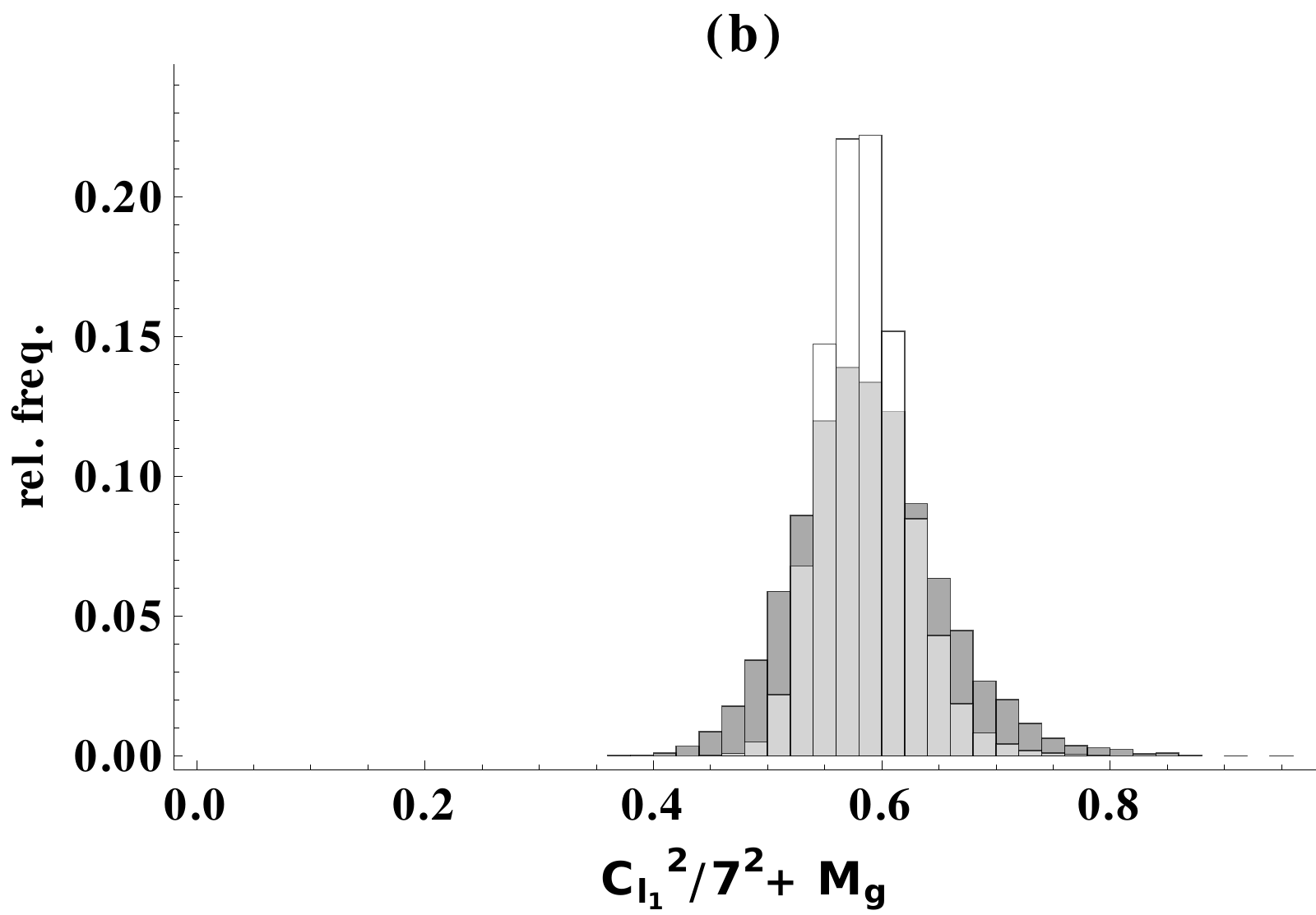}\\
\includegraphics[width=1.6in, angle=0]{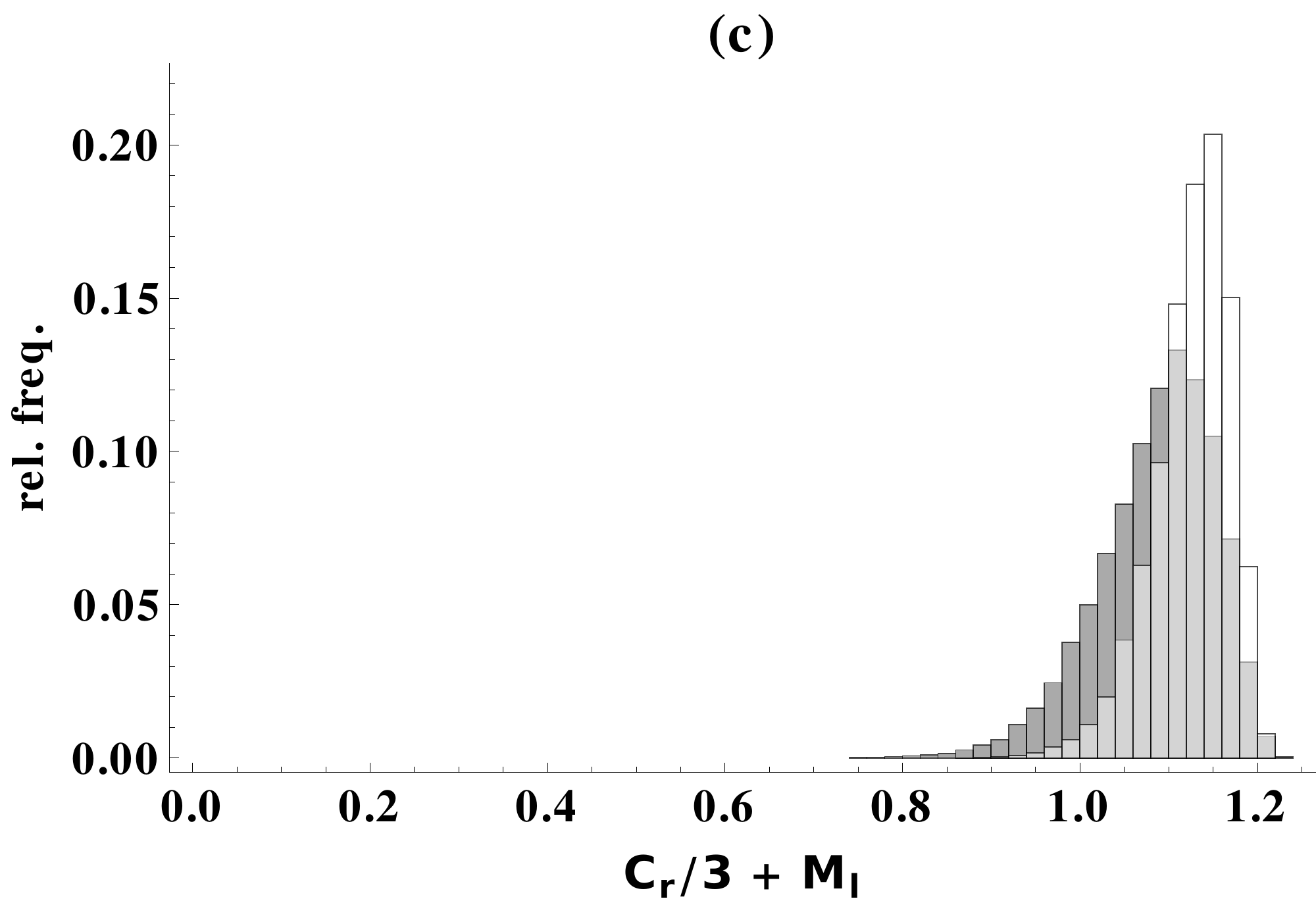} 
%\hspace{0.1cm}
\includegraphics[width=1.6in, angle=0]{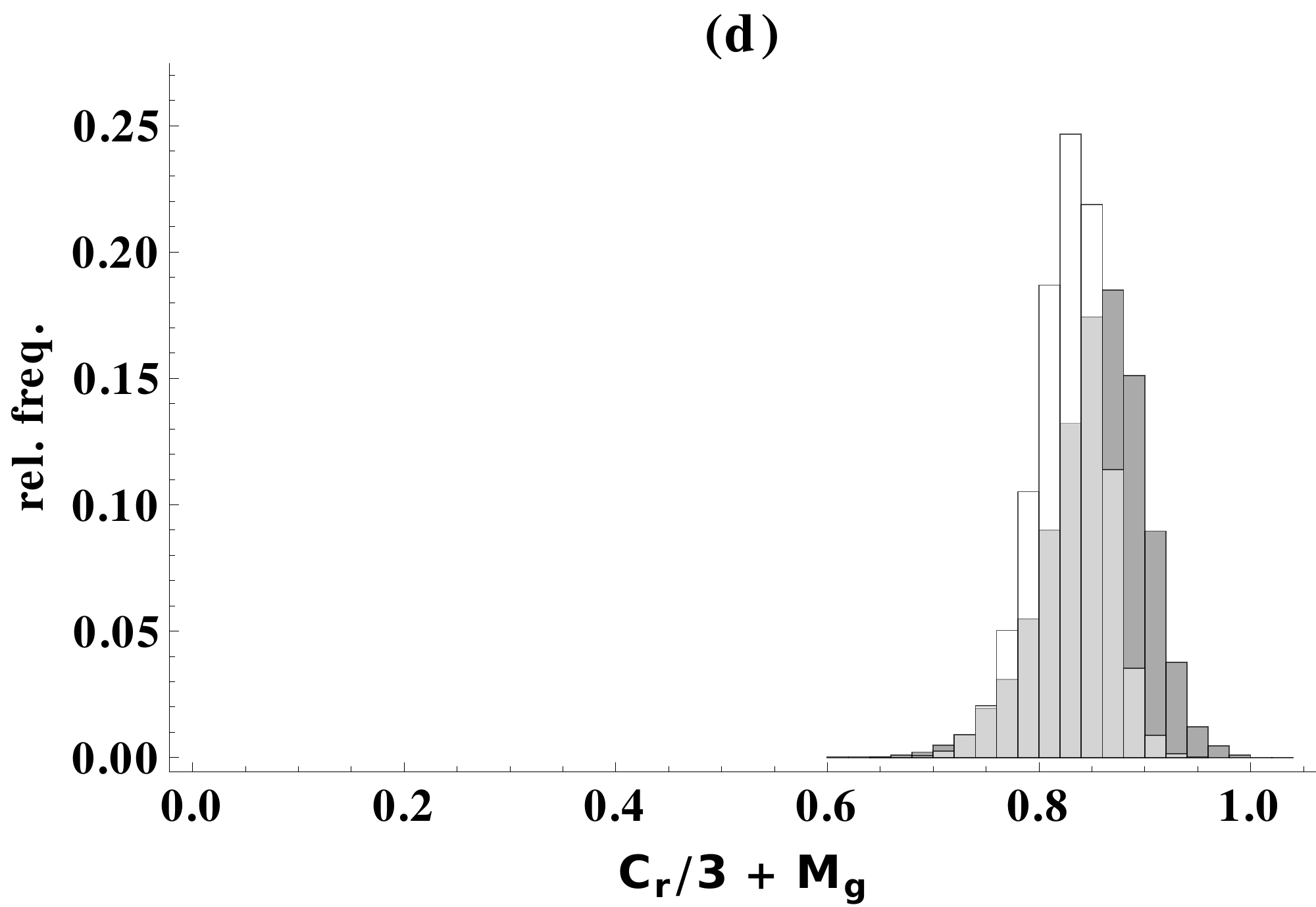}
\caption{(Color online.) Histograms depicting the relative frequency (rel. freq.) of quantum states against the trade-off between coherence and mixedness for different measures: (a) $\frac{C^2_{l_1}(\rho)}{(d-1)^2} + \frac{S(\rho)}{\mbox{ln}~d} {\leq} 1$, (b) $\frac{C^2_{l_1}(\rho)}{(d-1)^2} + M_g(\rho) {\leq} 1$, (c) $\frac{C_{r}(\rho)}{\mbox{ln}~d} + M_l(\rho) {\leq} 1$, and (d) $\frac{C_{r}(\rho)}{\mbox{ln}~d} + M_g(\rho) {\leq} 1$. For both rank-2 (gray bars) and rank-3 (white bars), $2\times 10^4$ {\bf three-qubit} states are generated Haar uniformly in the computational basis. We see that only the reciprocity relation, $\frac{C_{r}(\rho)}{\mbox{ln}~d} + M_l(\rho) {\leq} 1$, is in conflict. However, it is well below the trivial value 2, for all states. Higher is the rank of quantum states, more is the violation. 
%Both axes are dimensionless.
}
\label{fig:coh-mix-reciprocity-3q}
\end{figure}
\end{center}

\begin{center}
\begin{figure}[htb]
\includegraphics[width=1.6in, angle=0]{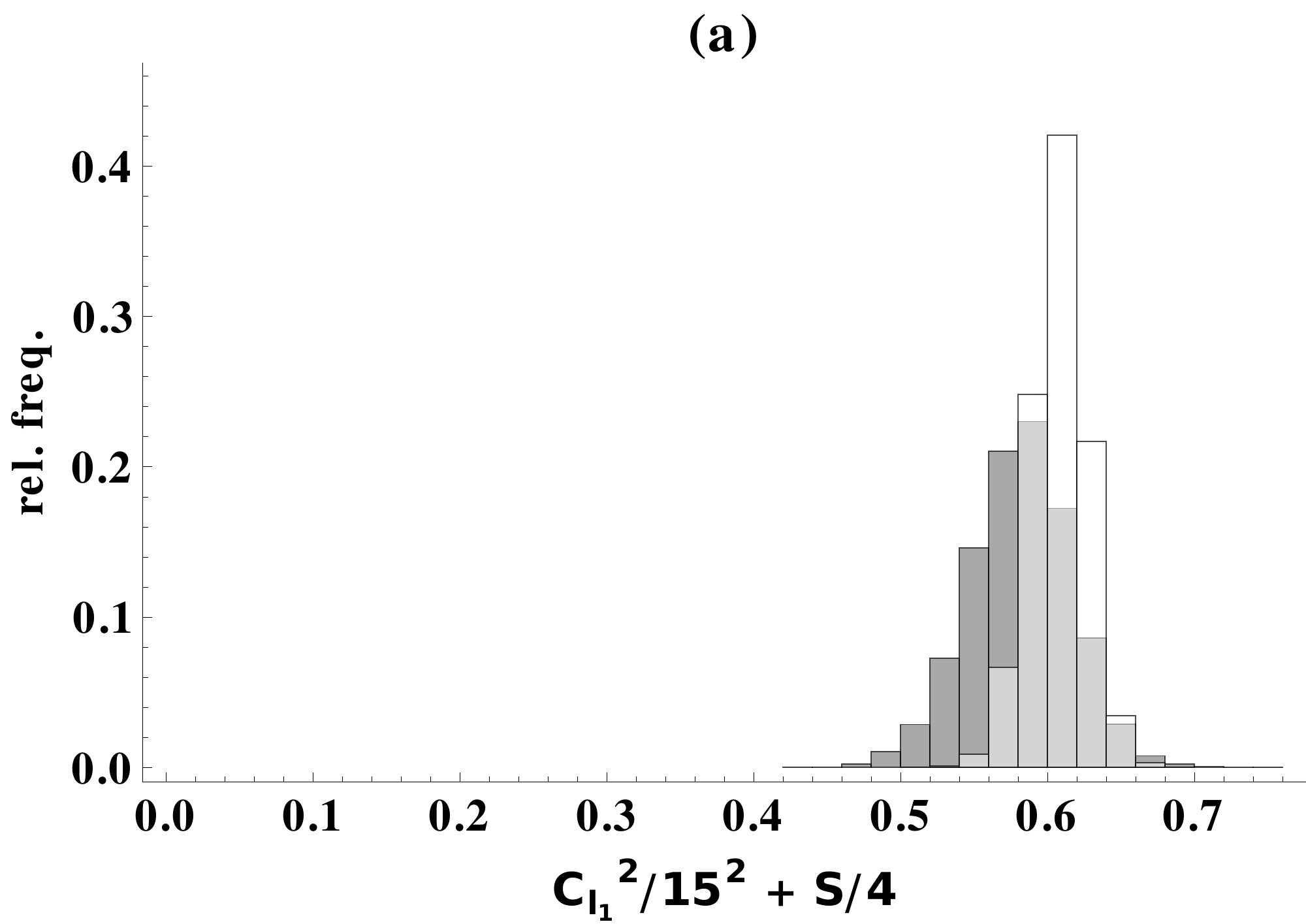} 
%\hspace{0.1cm}
\includegraphics[width=1.6in, angle=0]{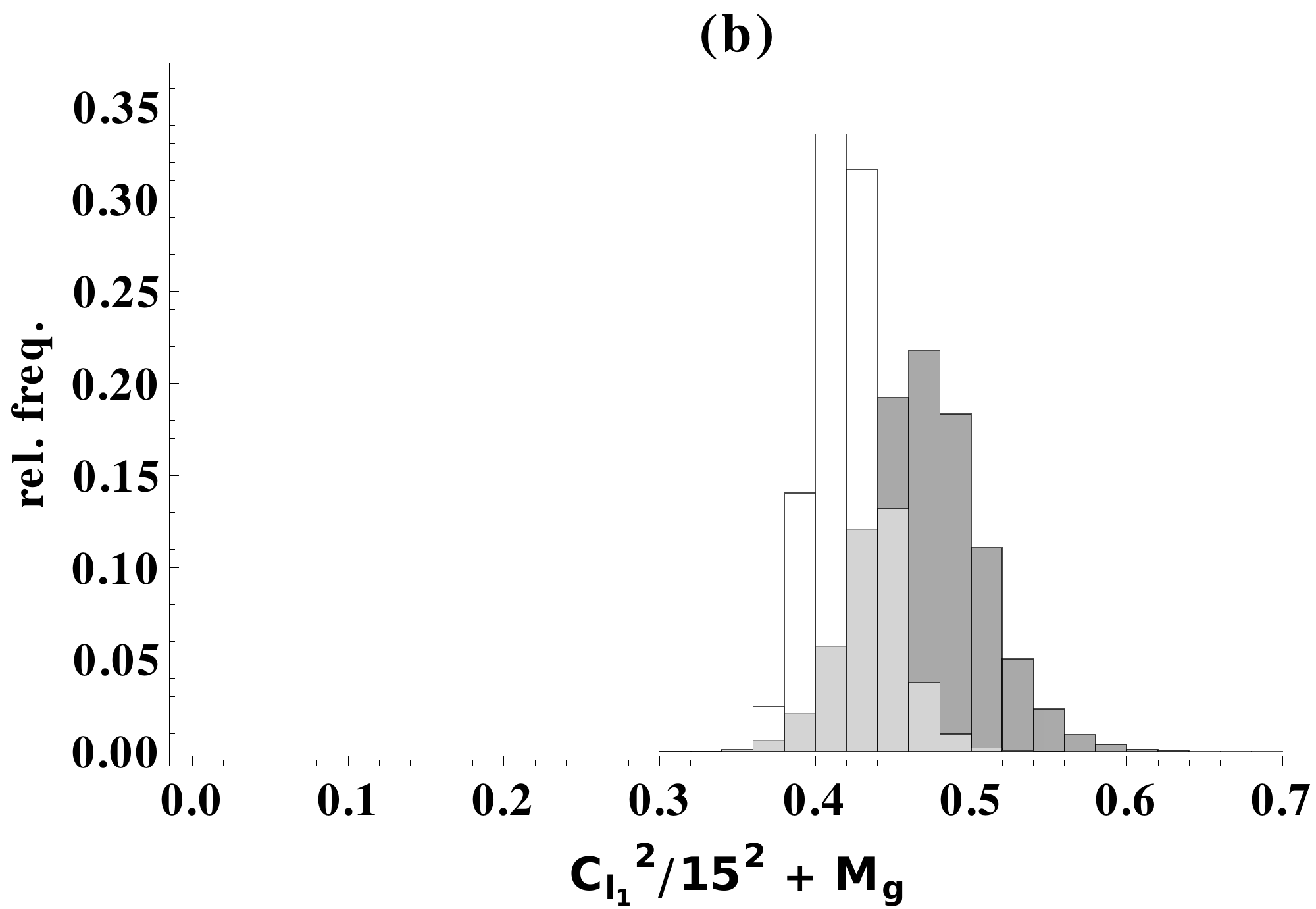}\\
\includegraphics[width=1.6in, angle=0]{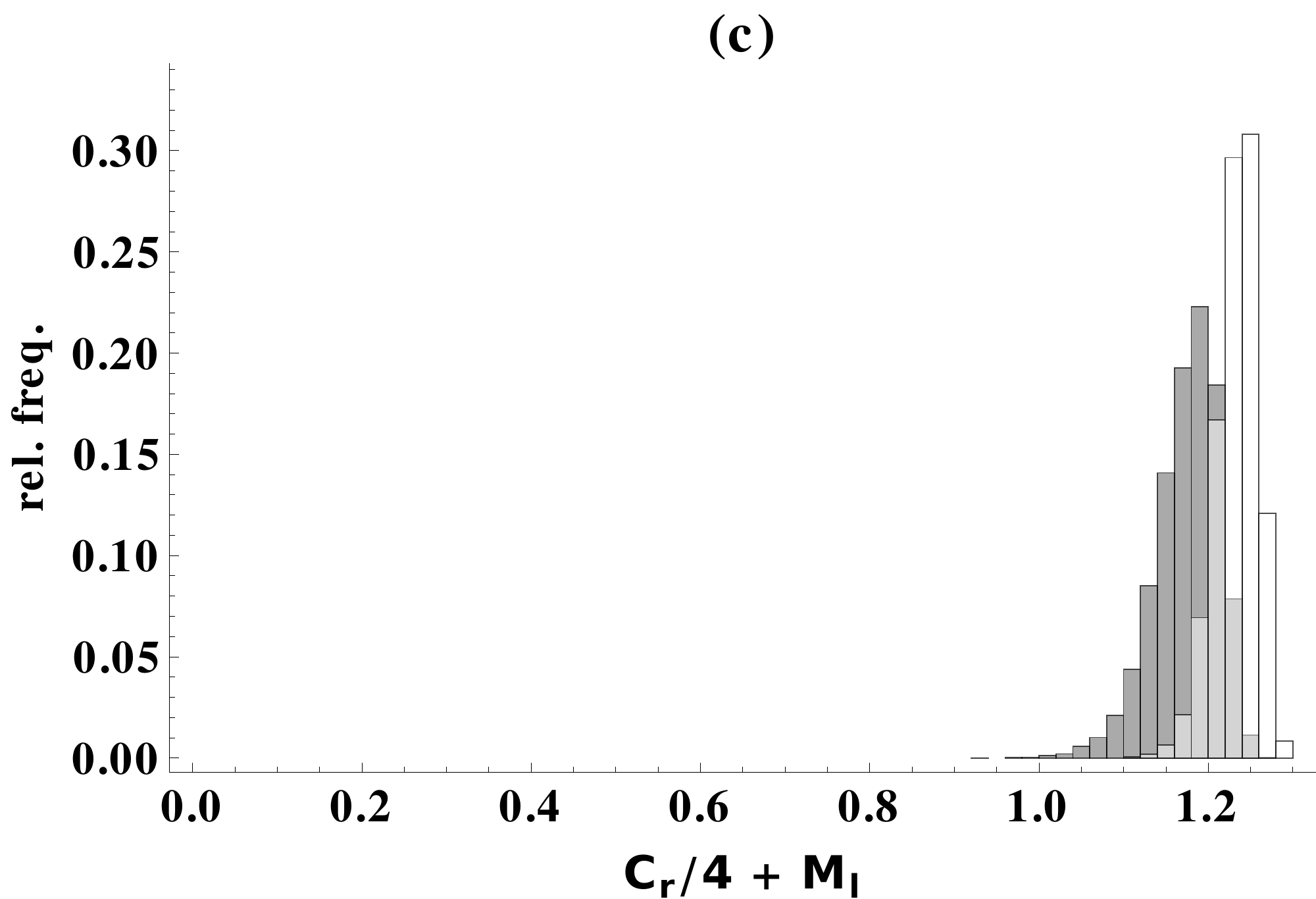} 
%\hspace{0.1cm}
\includegraphics[width=1.6in, angle=0]{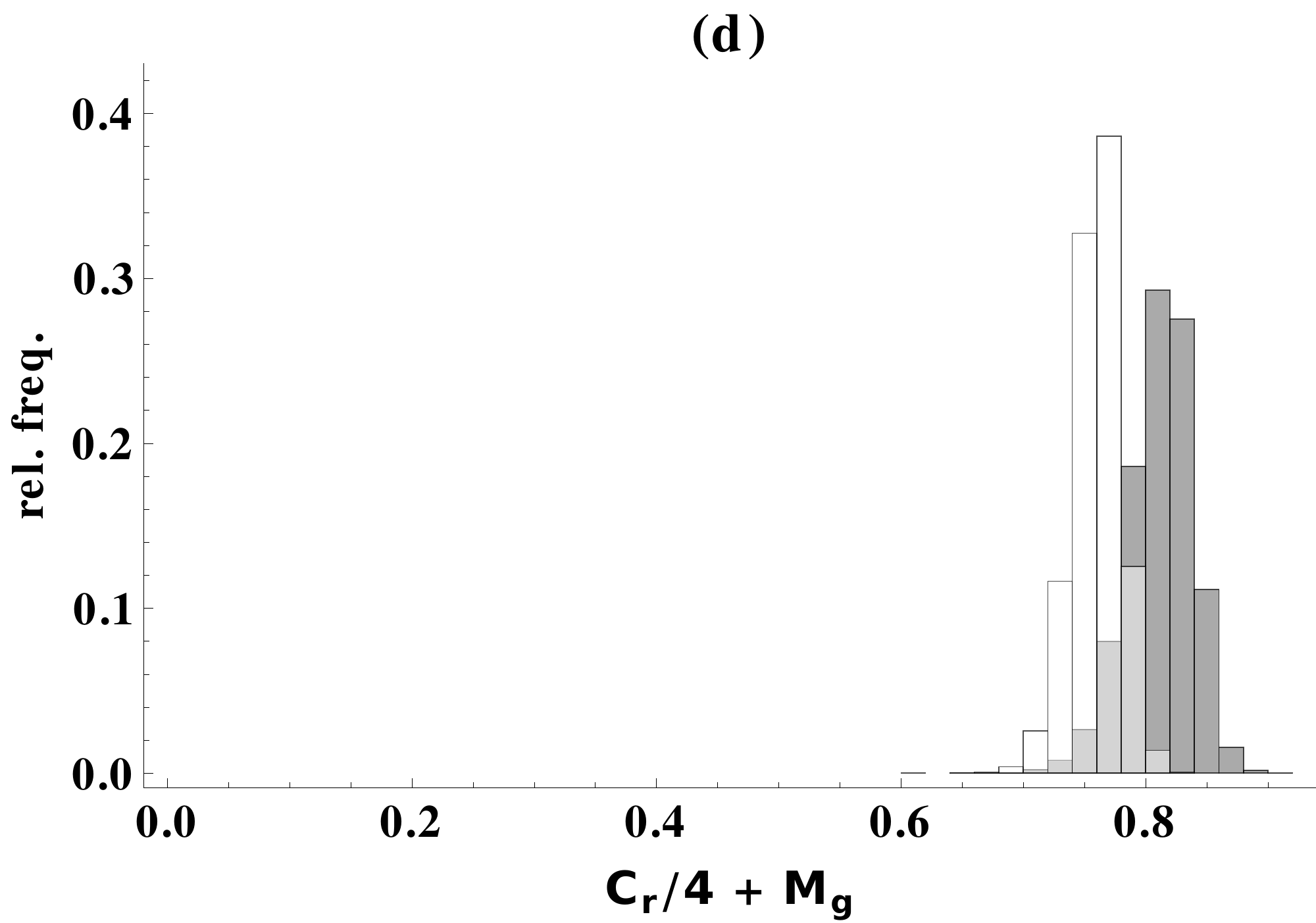}
\caption{(Color online.) Histograms depicting the relative frequency (rel. freq.) of quantum states against the trade-off between coherence and mixedness for different measures for rank-2 (gray bars) and rank-3 (white bars) {\bf four-qubit} states. Rest of the details are the same as in Fig. \ref{fig:coh-mix-reciprocity-3q}. 
%Trade-off between coherence and mixedness for different measures: (a) $\frac{C^2_{l_1}(\rho)}{(d-1)^2} + \frac{S(\rho)}{\mbox{ln}~d} {\leq} 1$, (b) $\frac{C^2_{l_1}(\rho)}{(d-1)^2} + M_g(\rho) {\leq} 1$, (c) $\frac{C_{r}(\rho)}{\mbox{ln}~d} + M_l(\rho) {\leq} 1$, and (d) $\frac{C_{r}(\rho)}{\mbox{ln}~d} + M_g(\rho) {\leq} 1$. For both rank-2 (lower blue dots) and rank-3 (upper maroon dots), $2\times 10^4$ {\bf four-qubit} states are generated Haar uniformly in the computational basis. We see that only the reciprocity relation, $\frac{C_{r}(\rho)}{\mbox{ln}~d} + M_l(\rho) {\leq} 1$, is in conflict. Higher is the rank of quantum states, more is the violation.  
%Both axes are dimensionless.
}
\label{fig:coh-mix-reciprocity-4q}
\end{figure}
\end{center}
%%%%

\section{Distribution of Quantum Coherence}
\label{coh-mono}
Quantum coherence is a resource. Coherence of a multiparty quantum system $\rho_{AB_1B_2\cdots B_n}$ is seen as a quantum correlation amongst the subsystems. We wish to study the distribution of coherence among the constituent subsystems. In particular, we are interested in the following monogamy-type relation \cite{ckw-mono}, which we refer to as additivity relation:
\begin{align}
C(\rho_{AB_1B_2\cdots B_n}) - \sum_{k=1}^n C(\rho_{AB_k})  \geq 0 ,
\label{eq:coh-mono}
\end{align}
where $C$ is some valid coherence measure and $\rho_{AB_k}$ is the two-party reduced density matrix obtained after partial tracing all subsystems but subsystems \(A\) and \(B_k\). If the above relation is satisfied for arbitrary quantum system $\rho_{AB_1B_2\cdots B_n}$, we say that the distribution of coherence is ``faithful'' with respect to subsystem \(A\). In the language of monogamy \cite{ckw-mono}, $C$ is monogamous with respect to pivot \(A\). If the above relation does not hold for any $\rho_{AB_1B_2\cdots B_n}$, $C$ is unfaithful or non-monogamous. Below we provide several interesting results on the distribution of quantum coherence in multipartite quantum systems. Remember that coherence is a basis-dependent quantity. In considering following results and theorems, we assume that quantum system under investigation is described in a fixed reference basis. Let $\{|a_i\rangle \}$ and 
$\{|b^{(k)}_j\rangle \}$ be the bases of subsystems \(A\) and \(B_k\) respectively, such that 
$|\psi\rangle_{AB} \equiv |\psi\rangle_{AB_1\cdots B_n}=\sum c_{a_ib_{k_1}\cdots b_{k_n}} 
|a_ib^{(1)}_{k_1}\cdots b^{(n)}_{k_n}\rangle$, $\rho_{AB}=\sum_{i} p_i
|\psi^i\rangle_{AB}\langle \psi^i|$, $\rho_{AB_j}=\mbox{Tr}_{\overline{AB_j}} (\rho_{AB})$, etc.

\subsection{Numerical results}
From numerical findings listed in Table \ref{table:coh-mono}, we observe a number of important results. The percentage of quantum states satisfying the coherence additivity relation increases with increasing number of parties, the rank of quantum states and raising the power of coherence measures under investigation. Furthermore, for fixed rank and fixed number of qubits, the number of quantum states which satisfy the monogamy condition is larger for entropy-based coherence measure than distance based coherence measure. 

\begin{center}
\begin{table}[htb]
\begin{tabular}{|c|c|c|c|c|c|c|}
\hline
{$\mbox{rank}$} & $\mbox{no. of qubits}$ & $\delta_{{\cal C}_{l_1}}$ & $\delta_{{\cal C}^2_{l_1}}$ & $\delta_{{\cal C}^3_{l_1}}$ & $\delta_{{\cal C}_{r}}$ & $\delta_{{\cal C}^2_{r}}$ \\ \hline \hline
                                     & 3                & 0.185                     & 32.045                      & 62.915                      & 5.14                    & 84.56                     \\ \cline{2-7} 
                                     & 4                & 0.015                     & 64.765                      & 94.445                      & 64.225                  & 99.92                     \\ \cline{2-7} 
\multirow{-3}{*}{1}                  & 5                & 0.035                     & 96.07                       & 99.95                       & 99.02                   & 100.                      \\ \hline \hline
                                     & 3                & 0.445                     & 38.245                      & 70.935                      & 75.425                  & 99.685                    \\ \cline{2-7} 
                                     & 4                & 0.095                     & 75.705                      & 97.74                       & 99.245                  & 100.                      \\ \cline{2-7} 
\multirow{-3}{*}{2}                  & 5                & 0.145                     & 98.715                      & 99.995                      & 99.995                  & 100.                      \\ \hline \hline
                                     & 3                & 0.615                     & 41.77                       & 73.885                      & 93.595                  & 99.98                     \\ \cline{2-7} 
                                     & 4                & 0.14                      & 79.475                      & 98.395                      & 99.975                  & 100.                      \\ \cline{2-7} 
\multirow{-3}{*}{3}                  & 5                & 0.185                     & 99.205                      & 100.                        & 100.                    & 100.                      \\ \hline \hline
                                     & 3                & 0.72                      & 42.385                      & 75.155                      & 97.                     & 99.985                    \\ \cline{2-7} 
                                     & 4                & 0.18                      & 80.845                      & 98.825                      & 100.                    & 100.                      \\ \cline{2-7} 
\multirow{-3}{*}{4}                  & 5                & 0.265                     & 99.385                      & 99.995                      & 100.                    & 100.                      \\ \hline
\end{tabular}
\caption{Percentage of quantum states, of varying ranks, satisfying the additivity relation for the ``normalized'' coherence measures ${\cal C}_{l_1}(\rho)$ and ${\cal C}_{r}(\rho)$ for \(3\), \(4\) and \(5\) qubits in the computational basis \cite{comment1}. The percentage of quantum states satisfying the additivity relation increases with increasing number of parties \cite{comment2}, with increment in the rank of quantum states, and with raising of the power of coherence measures under investigation. For every rank, $2\times 10^4$ three, four and five qubit states are generated Haar uniformly.}
\label{table:coh-mono}
\end{table}
\end{center}

\subsection{Analytical results}
In this section, we provide conditions for the violation of the additivity relation of the relative entropy of coherence $C_r$ \cite{note}. %We also prove that a measure of coherence that satisfies the additivity relation does satisfy the same on raising its power, and a measure of coherence that violates the additivity relation does violate the same on lowering its power.

\begin{thm}
\label{t1}
The relative measure of coherence violates the additivity relation for quantum states $\rho_{AB} \equiv \rho_{AB_1B_2\cdots B_n}$ which satisfy $S(\rho_{AB}) + S(\rho_A)=\sum_{k=1}^n S(\rho_{AB_k})$.
\end{thm}
%%%%
\begin{proof}
Let $\rho_{AB} \equiv \rho_{AB_1B_2\cdots B_n}$ be the density matrix of an $(n+1)$-party quantum system, and $\rho_{AB_k}=\mbox{Tr}_{\overline{AB_k}}\rho_{AB_1B_2\cdots B_n}$ be the reduced density matrix obtained after partial tracing all subsystems but subsystems \(A\) and \(B_k\). Then
\begin{eqnarray}
&& C_r(\rho_{AB})-\sum_{k=1}^n C_r(\rho_{AB_k}) \nonumber \\
& = & S(\rho ^I_{AB})-S(\rho_{AB})-\sum_{k=1}^n \left[S(\rho ^I_{AB_k})-S(\rho_{AB_k})\right] \nonumber \\
& = & \left[\sum_{k=1}^n S(\rho_{AB_k})-S(\rho_{AB})-S(\rho_A)\right]\nonumber \\ 
& - & \left[\sum_{k=1}^n S(\rho ^I_{AB_k})-S(\rho ^I_{AB})-S(\rho ^I_A)\right]
- \left[S(\rho ^I_A)-S(\rho_A)\right]\nonumber \\
& = & \Delta_{1} - \Delta_{2 } - C_r(\rho_A).
\end{eqnarray}
%where $\Delta_{1}, ~\Delta_{2}$ and $\Delta_{3}$ represent the quantities inside the square brackets in the last line, respectively.
It can be easily shown that for $\sigma_{AB} \equiv \sigma_{AB_1B_2\cdots B_n}$, $\sum_{k=1}^n S(\sigma_{AB_k})-S(\sigma_{AB}) \geq (n-1)S(\sigma_A) \geq 0$. This bound is a simple consequence of the strong subadditivity relation,  $S(\rho_{ABC})+S(\rho_A) \leq S(\rho_{AB})+S(\rho_{AC})$, of von Neumann entropy. Hence $\Delta_{1}$ and $\Delta_{2}$ are non-negative. 
%Also $\Delta_{3}$ is non-negative, by definition of the relative entropy of coherence. 
When $\sum_{k=1}^n S(\rho_{AB_k})=S(\rho_{AB}) + S(\rho_A)$, i.e., $\Delta_{1}$ vanishes, $C_r(\rho_{AB}) \leq \sum_{k=1}^n C_r(\rho_{AB_k})$. Thus $C_r$ violates the additivity relation.\\
%%%
Very recently, a special case of above result was obtained in Ref. \cite{yao}. It was shown that $C_r$ violates the additivity relation for an arbitrary tripartite state $\rho_{ABC}$ which saturates the strong subadditivity relation %$S(\rho_{ABC})+S(\rho_A) \leq S(\rho_{AB})+S(\rho_{AC})$, 
of von Neumann entropy. Tripartite states satisfying the strong subadditivity relation are reported in Ref. \cite{hpw246}.\\ 

Analogous result can be obtained for other distributions as well. For instance, the following distribution yields 
\begin{eqnarray}
&& C_r(\rho_{AB})-\sum_{k=1}^n C_r(\rho_{A\overline{B_k}}) \nonumber \\
& = & S(\rho ^I_{AB})-S(\rho_{AB})-\sum_{k=1}^n \left[S(\rho ^I_{A\overline{B_k}})-S(\rho_{A\overline{B_k}})\right] \nonumber \\
& = & \left[S(\rho ^I_{AB})-\sum_{k=1}^n S(\rho ^I_{A\overline{B_k}})\right] + \left[\sum_{k=1}^n S(\rho_{A\overline{B_k}})-S(\rho_{AB})\right] \nonumber \\
& = & \left[\sum_{k=1}^n S(\rho_{A\overline{B_k}})-(n-1)S(\rho_{AB})\right] - \left[\sum_{k=1}^n S(\rho ^I_{A\overline{B_k}})-(n-1)S(\rho ^I_{AB})\right] \nonumber \\
& - & (n-2)\left[S(\rho ^I_{AB})-S(\rho_{AB})\right] \nonumber \\
& = & \Delta_{3} - \Delta_{4} - (n-2)C_r(\rho_{AB}),
\end{eqnarray}
where $\rho_{A\overline{B_k}}=\mbox{Tr}_{B_k}\rho_{AB_1B_2\cdots B_n}$ be the reduced density matrix obtained after partial tracing subsystem \(B_k\).
Again since $\sum_{k=1}^n S(\sigma_{A\overline{B_k}})-(n-1)S(\sigma_{AB}) \geq S(\sigma_A) \geq 0$ for $\sigma_{AB} \equiv \sigma_{AB_1B_2\cdots B_n}$ \cite{asu-qmi, note1}, $\Delta_{3}$ and $\Delta_{4}$ are non-negative. When either $\Delta_{3} \leq \Delta_{4} + (n-2)C_r(\rho_{AB})$ or $\Delta_{3}=0$, $C_r(\rho_{AB}) \leq \sum_{k=1}^n C_r(\rho_{A\overline{B_k}})$. Thus $C_r$ violates the additivity relation.  
\end{proof}
%%%%%

However, coherence measures are not normalized in general. That is, they do not lie between zero and unity for arbitrary quantum systems. But in investigating monogamy relations for quantum correlation measures, we consider that value of all quantities in the monogamy inequality lies in the same range. Therefore, it is reasonable to consider the normalized coherence. Suppose that $\rho_{AB} \equiv \rho_{AB_1B_2\cdots B_n} \in 
\left(\mathbb{C}^d\right)^{\otimes (n+1)}$ be a multipartite density operator. Considering the normalized relative entropy of coherence, we have
\begin{eqnarray}
&& \frac{C_r(\rho_{AB})}{\mbox{ln}~d^{n+1}}-\sum_{k=1}^n \frac{C_r(\rho_{AB_k})}{\mbox{ln}~d^2} \nonumber \\
& = & \frac{2(\Delta_{1} - \Delta_{2})-(n-1)\sum_{k=1}^n C_r(\rho_{AB_k})}{2(n+1)\mbox{ln}~d}.
\end{eqnarray}
Again, when $\Delta_{1}$ vanishes, the normalized $C_r$ does not satisfy the additivity relation. Similarly, for the other distribution, we can obtain
\begin{eqnarray}
&& \frac{C_r(\rho_{AB})}{\mbox{ln}~d^{n+1}}-\sum_{k=1}^n \frac{C_r(\rho_{A\overline{B_k}})}{\mbox{ln}~d^{n}}  \nonumber \\
& = & \frac{n(\Delta_{3} - \Delta_{4}) - \sum_{k=1}^n C_r(\rho_{A\overline{B_k}}) - n(n-2)C_r(\rho_{AB})}{n(n+1)\mbox{ln}~d}.
\end{eqnarray}
Thus, when $\Delta_{3}=0$, the normalized $C_r$ violates the additivity relation.\\

\section{Coherence in X states}
\label{x-state}
Quantum states having ``X''-structure are referred to as X states. Consider an \(n\)-qubit X state given by
\begin{equation}
\rho = p|gGHZ\rangle\langle gGHZ|+(1-p)\frac{I_d}{d},
\label{eq:XXstate}
\end{equation}
where $|gGHZ\rangle=(\alpha|0\rangle^{\otimes n}+\beta|1\rangle^{\otimes n})$ with $|\alpha|^2+|\beta|^2$, $I_d$ is $d\times d$ identity matrix, $d=2^n$ and $0\leq p \leq 1$. It is easy to show that for this state, $C_{l_1}(\rho)=2p|\alpha \bar{\beta}|$ and 
$C_r(\rho)=-\left(p|\alpha|^2+\frac{1-p}{d}\right)\log_2\left(p|\alpha|^2+\frac{1-p}{d}\right)
-\left(p|\beta|^2+\frac{1-p}{d}\right)\log_2\left(p|\beta|^2+\frac{1-p}{d}\right)
-\frac{1-p}{d}\log_2 \frac{1-p}{d}-\left(p+\frac{1-p}{d}\right)\log_2\left(p+\frac{1-p}{d}\right)$.
%%%%
\begin{thm}
\label{t9}
For an \((n+1)\)-party X state $\rho ^X_{AB_1B_2\cdots B_n}$ in a given basis, any measure of coherence $C$ satisfies the additivity relation. That is, X states satisfy the relation $C(\rho ^X_{AB_1B_2\cdots B_n}) - \sum_{k=1}^n C(\rho ^X_{AB_k}) \geq 0$, where $\rho ^X_{AB_k}=\mbox{Tr}_{\overline{AB_k}}\rho ^X_{AB_1B_2\cdots B_n}$ is a two-qubit reduced density matrix. 
\end{thm}
%%%%%
\begin{proof}
This is because all the two-party reduced density matrices of $(n+1)$-party X states in the given basis are diagonal and any valid measure of quantum coherence vanishes for diagonal states. 
\end{proof}
%%%
\begin{center}
\begin{figure}[htb]
\includegraphics[width=3.2in, angle=0]{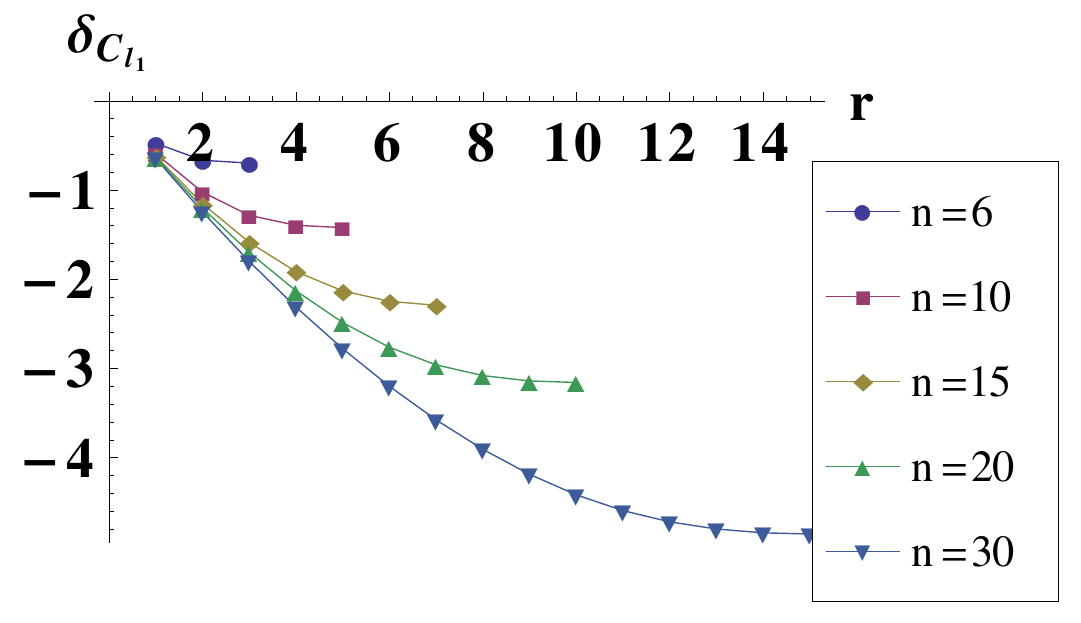}
\includegraphics[width=3.2in, angle=0]{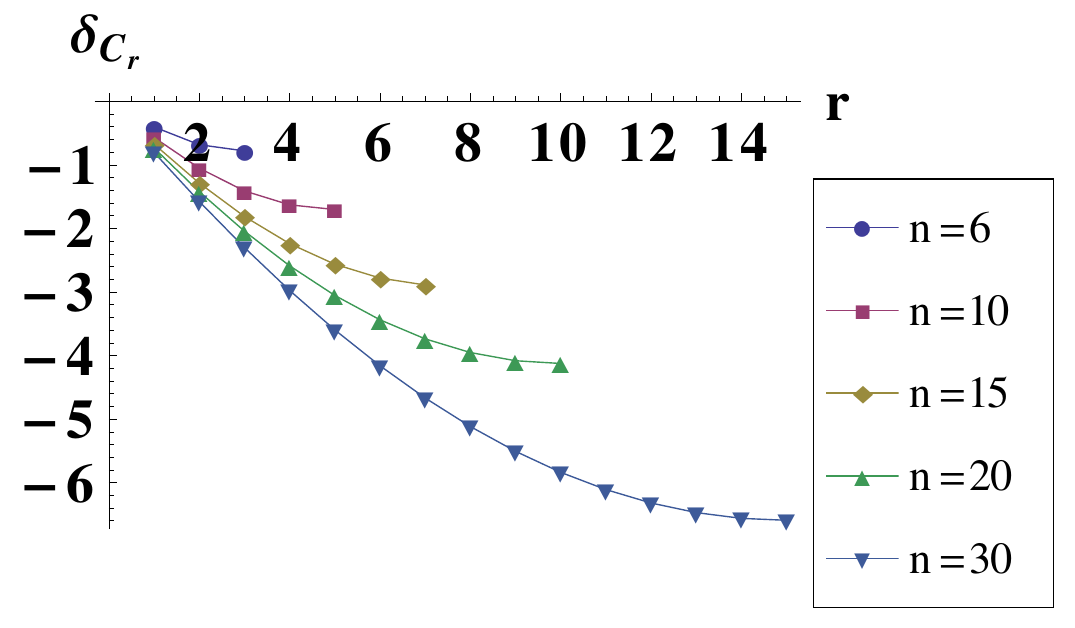}
\caption{Coherence score (y-axis) versus the number of excitations $r \leq \frac{n}{2}$ (x-axis) of Dicke states using the ``normalized'' coherence measures $C_{l_1}$ (top panel) and $C_r$ (bottom panel). All quantities are dimensionless. The normalized measures of coherence do not satisfy the additivity relation of coherence for the Dicke states.}
\label{fig:dicke-coh-score}
\end{figure}
\end{center}

Now consider the Dicke states \cite{dickestates}, which is symmetric with respect to interchange of qubits, given by
\begin{equation}
|D_{n,r}\rangle=\binom{n}{r} ^{-\frac12} \sum_{{\cal P}} {\cal P}\left(|0\rangle^{\otimes (n-r)}\otimes 
|1\rangle^{\otimes r}\right),
\label{eq:Dicke}
\end{equation}
where $\sum_{{\cal P}}$ represents sum over all $\binom{n}{r}$ permutations of $(n-r)$ $|0\rangle$s and 
$r$ $|1\rangle$s.
Note that the Dicke state itself, in Eq.~(\ref{eq:Dicke}), is not an X state but its all two-qubit reduced density matrices, in the computational basis, are same and has X structure. Again, one can show that for the normalized measures of coherence 
\begin{align}
\delta_{C^{(norm)}_{l_1}}\left(|D_{n,r}\rangle\right)&=\frac{C_{l_1}\left(|D_{n,r}\rangle\right)}{2^n-1}
-(n-1)\frac{C_{l_1}\left(\rho^{(2)}_{|D_{n,r}\rangle}\right)}{3} \nonumber \\
&=\frac{\binom{n}{r}-1}{2^n-1}-\frac{2r(n-r)}{3n},
\label{eq:dicke-mono-cp1-norm}
\end{align}

and 
\begin{align}
\delta_{C^{(norm)}_{r}}\left(|D_{n,r}\rangle\right)&=\frac{C_{r}\left(|D_{n,r}\rangle\right)}{\log_2 2^n}
-(n-1)\frac{C_{r}\left(\rho^{(2)}_{|D_{n,r}\rangle}\right)}{\log_2 4} \nonumber \\
&=\frac{1}{n}\log_2 \binom{n}{r}-\frac{r(n-r)}{n},
\label{eq:dicke-mono-cr-norm}
\end{align}
where $\rho^{(2)}_{|D_{n,r}\rangle}$ is the two-qubit reduced density matrix of the Dicke state.
For \(n \geq 3\) and $1 \leq r \leq n$, $\delta_{C_{l_1}}(|D_{n,r}\rangle)$ and 
$\delta_{C_{l_1}}(|D_{n,r}\rangle)$ are non-positive. Thus, quantum coherence measures violate the additivity relation for the Dicke states in the computational basis (see Fig. \ref{fig:dicke-coh-score}).\\
%%%%
Analogous result was obtained in Ref. \cite{apau15}, that the Dicke state is always non-monogamous with respect to quantum discord \cite{discord} and quantum work-deficit \cite{work-deficit}, and the Dicke state with more number of parties is more non-monogamous to that with a smaller number of parties.\\

However, if one considers the unnormalized measures of coherence, then
\begin{align}
\delta_{C_{l_1}}\left(|D_{n,r}\rangle\right)&=C_{l_1}\left(|D_{n,r}\rangle\right)
-(n-1)C_{l_1}\left(\rho^{(2)}_{|D_{n,r}\rangle}\right) \nonumber \\
&=\binom{n}{r}-1-\frac{2r(n-r)}{n},
\label{eq:dicke-mono-cp1}
\end{align}

and 
\begin{align}
\delta_{C_{r}}\left(|D_{n,r}\rangle\right)&=C_{r}\left(|D_{n,r}\rangle\right)
-(n-1)C_{r}\left(\rho^{(2)}_{|D_{n,r}\rangle}\right) \nonumber \\
&=\log_2 \binom{n}{r}-\frac{2r(n-r)}{n}.
\label{eq:dicke-mono-cr}
\end{align}
In this case, when \(n \geq 3\) and $1 \leq r \leq n$, $\delta_{C_{l_1}}(|D_{n,r}\rangle)$ and 
$\delta_{C_{l_1}}(|D_{n,r}\rangle)$ are non-negative. Thus, quantum (unnormalized) coherence measures satisfy the additivity relation for the Dicke states in the computational basis.

\section{Conclusion}
\label{discussion}

%\noindent\emph{Conclusion}.--
In this paper, we have shown that the reciprocity between coherence and mixedness of quantum states is a general feature as this complementarity holds for large spectra of measures of coherence and of mixedness.
%We have also studied the distribution of coherence in multipartite systems and proven several interesting results.
The numerical investigation of the distribution of coherence in multipartite systems reveals that the percentage of quantum states satisfying the additivity relation increases with increasing number of parties, with increment in the rank of quantum states, and with raising of the power of coherence measures under investigation. 
%Moreover, a monogamous measure of coherence remains monogamous on raising its power and a non-monogamous measure remains non-monogamous on lowering its power. 
We have provided conditions for the violation of the additivity relation of the relative entropy of coherence. 
%We have also shown analytically that a measure of coherence that satisfies the additivity relation does satisfy the same on raising its power, and a measure of coherence that violates the additivity relation does violate the same on lowering its power. 
We have further shown that the normalized measures of coherence violate the additivity relation for the Dicke states.\\

\begin{acknowledgments}
AK acknowledges a research fellowship of the Department of Atomic Energy, Government of India. The author is very grateful to Ujjwal Sen, Uttam Singh and Himadri Shekhar Dhar for useful comments and suggestions. 
%We acknowledge computations performed at cluster computing facilities at HRI.
\end{acknowledgments}


\begin{thebibliography}{99}
\bibitem{comment1} Though the chosen reference basis is a product basis, the states were generated Haar uniformly which consisted of non-local states also. That is, the states are generated randomly with an invariant distribution under unitary transformations. Hence we understand that the numerical results should be basis-independent.
\bibitem{comment2} Note that the percentage of quantum states satisfying the additivity relation does increase with increasing number of qubits (for number of qubits $4$ onwards) in the third column of Table \ref{table:coh-mono}. However, the reason behind the dip in percentage from three-qubit states to four-qubit states is not understood.
\bibitem{note} Several other analytical results can be shown to hold true for quantum coherence: (i) A measure of coherence $C$ that violates the additivity relation for a multipartite system, can be made to satisfy the same by raising its power to some positive real number ($>1$), provided $C$ is monotonically decreasing under discarding of subsystems. (ii) A measure of coherence that satisfies the additivity relation does satisfy the same on raising its power, and a measure of coherence that violates the additivity relation does violate the same on lowering its power. However, it is not known whether raising and lowering operations on bonafide measures of coherence will yield another valid coherence measures.
\end{thebibliography}
\end{document}